\documentclass[letterpaper]{article} 
\usepackage{aaai2026}  
\usepackage{times}  
\usepackage{helvet}  
\usepackage{courier}  
\usepackage[hyphens]{url}  
\usepackage{graphicx} 
\usepackage{natbib}  
\usepackage{caption} 
\usepackage{algorithm}

\usepackage{xcolor}
\usepackage{amssymb}
\usepackage{booktabs}
\usepackage{multirow}
\usepackage{makecell}
\usepackage{pifont}
\usepackage{amsmath}
\usepackage{amsthm}
\usepackage{subcaption}
\usepackage{bm}
\usepackage{microtype}
\newtheorem{theorem}{Theorem}

\usepackage{xr}
\usepackage[noend]{algpseudocode}

\usepackage{newfloat}
\usepackage{listings}
\DeclareCaptionStyle{ruled}{labelfont=normalfont,labelsep=colon,strut=off} 
\floatstyle{ruled}
\newfloat{listing}{tb}{lst}{}
\floatname{listing}{Listing}
\title{\textls[-1]{WeightFlow: Learning Stochastic Dynamics via Evolving Weight of Neural Network}}
\author {
    Ruikun Li\textsuperscript{\rm 1},
    Jiazhen Liu\textsuperscript{\rm 2},
    Huandong Wang\textsuperscript{\rm 2}\thanks{Corresponding author (wanghuandong@tsinghua.edu.cn).},
    Qingmin Liao\textsuperscript{\rm 1},
    Yong Li\textsuperscript{\rm 2}
}
\affiliations {
    \textsuperscript{\rm 1} Shenzhen International Graduate School, Tsinghua University\\
    \textsuperscript{\rm 2} Department of Electronic Engineering, BNRist, Tsinghua University\\
}

\begin{document}

\maketitle

\begin{abstract}
Modeling stochastic dynamics from discrete observations is a key interdisciplinary challenge. 
Existing methods often fail to estimate the continuous evolution of probability densities from trajectories or face the curse of dimensionality. 
To address these limitations, we presents a novel paradigm: modeling dynamics directly in the weight space of a neural network by projecting the evolving probability distribution.
We first theoretically establish the connection between dynamic optimal transport in measure space and an equivalent energy functional in weight space. 
Subsequently, we design WeightFlow, which constructs the neural network weights into a graph and learns its evolution via a graph controlled differential equation.
Experiments on interdisciplinary datasets show that WeightFlow improves performance by an average of 43.02\% over state-of-the-art methods, providing an effective and scalable solution for modeling high-dimensional stochastic dynamics.
\end{abstract}

\begin{links}
\link{Code}{https://github.com/tsinghua-fib-lab/WeightFlow}
\link{Page}{https://tsinghua-fib-lab.github.io/WeightFlow/}
\end{links}

\section{Introduction}

Stochastic dynamical processes are ubiquitous across numerous scales, from gene regulation and ecological evolution to climate patterns, making them a core subject of interdisciplinary research~\cite{sha2024reconstructing,wagner2023evolvability,lenton2019climate,li2025predicting3}.
A key challenge in modeling such systems lies in solving the temporal evolution of their state distribution, which is deterministically described by the renowned Fokker-Planck equation or the Master Equation~\cite{chen2017beating,jiang2021neural,liu2025beyond}.
Consequently, early efforts focused on developing efficient solution algorithms, including statistical approximations and variational autoregressive networks~\cite{chen2018efficient,anderson2024fisher,tang2023neural,liu2024distilling}. 
These studies, however, are predicated on the availability of well-defined governing equations or reaction rules. In real-world scenarios, we often lack this prior knowledge and only have access to static snapshots of the system at discrete time points~\cite{gao2022unitvelo,neklyudov2023action,qu2025gene}.

\begin{figure}[!ht]
    \centering
    \includegraphics[width=1\linewidth]{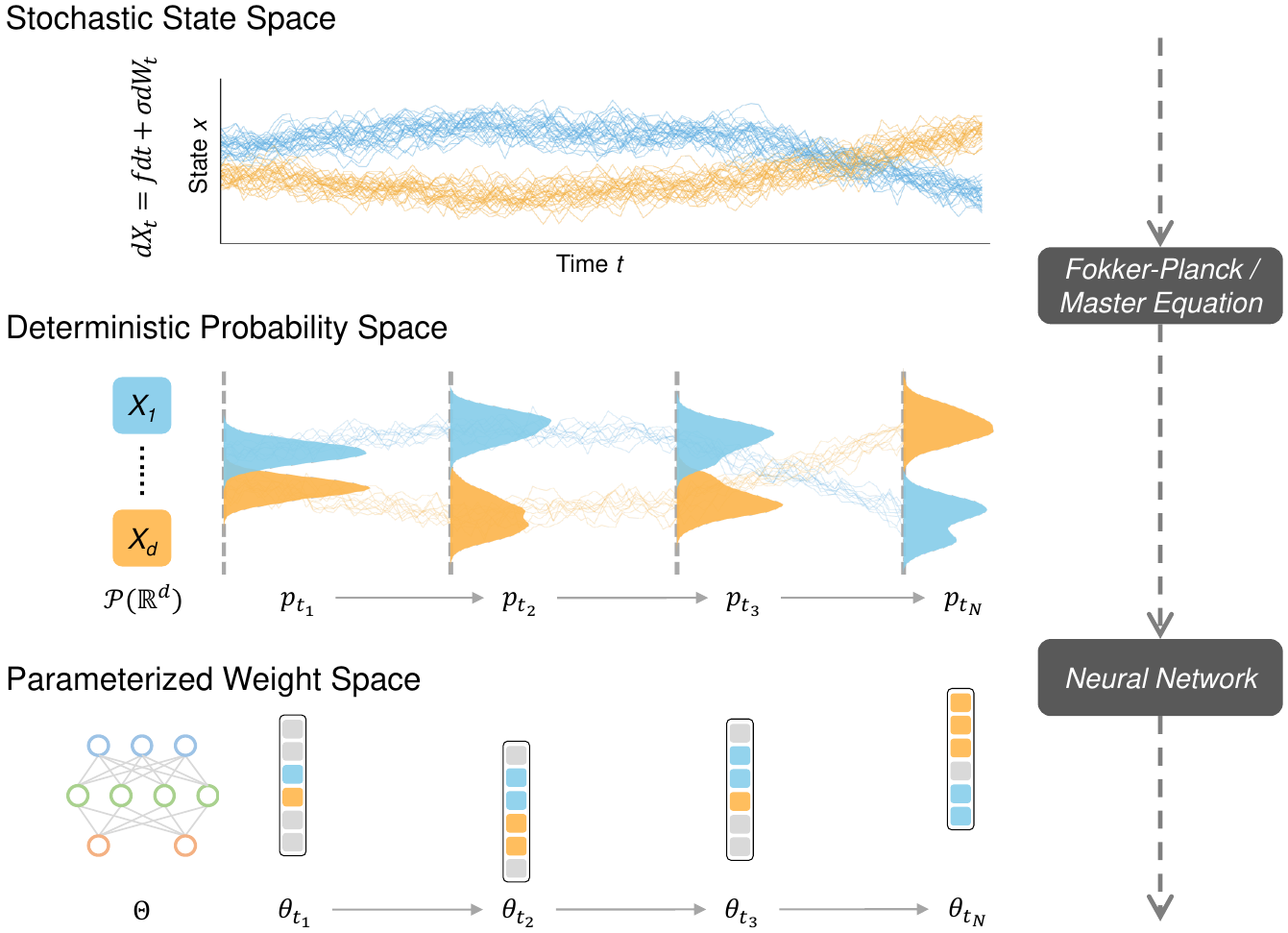}
    \caption{Learning stochastic dynamics: from sample trajectories to evolving weights.}
    \label{fig:intro}
\end{figure}

In the absence of governing equations, existing methods for inferring continuous dynamics from observational data can be classified into two categories. 
The first category models the stochastic trajectories of individual particles in state space.
For instance, NeuralSDE~\cite{li2020scalable} and SPINN~\cite{zhang2025physics} integrate stochastic differential equations to fit the drift and diffusion mechanisms of particle motion. 
Methods based on dynamic optimal transport are inspired by the principle of least action and construct trajectories by minimizing path energy~\cite{tong2020trajectorynet,onken2021ot,huguet2022manifold}. 
More recent studies~\cite{koshizuka2023neural,chen2023deep,kapusniak2024metric,terpin2024learning} have formalized this process as a Schrödinger Bridge problem, ensuring the physical realism and smoothness of trajectories by learning a potential function or introducing momentum.
A common limitation of these trajectory-based approaches is their inability to directly yield the probability density of states. 
Instead, the distribution must be approximated through large-scale trajectory sampling, which makes it exceedingly difficult to estimate rare events or the tails of the distribution.
To address this, a second category of methods emphasizes the evolution of deterministic probability distributions.
For example, transfer operator-based methods~\cite{mardt2018vampnets,schreiner2023implicit,kostic2023learning,federici2024latent} learn the transition probabilities of the system between metastable states. 
NeuralMJP~\cite{seifner2023neural} and ANN-SM~\cite{jiang2021neural} explicitly learn the master equation for discrete systems, thereby directly modeling the temporal evolution of the state distribution. 
Similarly, PESLA~\cite{li2025predicting} and DeepRUOT~\cite{zhang25learning} parameterize and solve the Fokker-Planck equation for continuous state spaces. 
However, these methods require describing a state space that grows exponentially with the system's dimension~\cite{tang2023neural}.
They not only face the curse of dimensionality during optimization but also struggle to capture the underlying geometric manifold~\cite{knobloch1983bifurcations} of the dynamical process.

In light of the limitations of existing methods, we propose a novel approach: projecting the evolution of the state distribution, $p(x,t)$, into the weight space of a neural network that defines it, $\Theta=\{\theta | p(x,t)=p_{\theta_t}(x)\}$, and directly modeling the dynamics of weight evolution (Figure~\ref{fig:intro}).
Compared to probability measures unfolded in high-dimensional spaces, the weights benefit from a well-defined topological structure imposed by the neural network's connectivity, making them amenable to modern graph representation models~\cite{li2025predicting2}.
Furthermore, the number of parameters in neural networks designed for sequential modeling, such as the autoregressive network~\cite{tang2023neural, feng2025singer}, does not depend directly on the system's dimension, thus effectively circumventing the curse of dimensionality.

In this work, we first theoretically demonstrate that the dynamic optimal transport path in probability space can be approximated by solving an equivalent energy functional in the weight space. 
Building on this, we design a novel framework named WeightFlow. 
WeightFlow models the neural network weights as a graph and employs a graph neural differential equation to learn the continuous dynamics of this weight graph. 
To capture the global manifold of dynamics, WeightFlow projects a latent path from observational snapshots and performs a Riemann–Stieltjes integral to achieve continuous interpolation. 
Our contributions are as follows:
\begin{itemize}
    \item We propose a novel framework, WeightFlow, for modeling stochastic dynamics in the weight space of a neural network, offering a new paradigm for tackling the challenge of modeling high-dimensional stochastic systems.
    \item We theoretically derive the approximately equivalent formulation in neural network weight space for the dynamic optimal transport path found in probability space.
    \item WeightFlow models network weights as a weight graph and introduces a graph controlled differential equation to learn its continuous evolution.
    \item WeightFlow achieves an average performance improvement of 43.02\% over existing state-of-the-art methods. Code is available at https://github.com/tsinghua-fib-lab/WeightFlow.
\end{itemize}

\section{Preliminary}

\subsection{Stochastic Dynamics \& Probability Evolution}
We consider a physical system whose state $X_t \in \mathbb{R}^d$ evolves according to a stochastic differential equation:
\begin{equation}
    dX_t = f(X_t, t)dt + \sigma(X_t, t)dW_t
\end{equation}
where $W_t$ is a standard Wiener process, $f(X_t, t)$ is the drift, and $\sigma(X_t, t)$ is the diffusion tensor.
While individual trajectories are stochastic, the ensemble's probability density $p(x,t)$ evolves deterministically according to the Fokker-Planck equation~\cite{risken1989fokker,gardiner2009stochastic}:
\begin{equation}
    \partial_t p(\bm{x},t) = -\nabla_x \cdot [(f(\bm{x},t)p(\bm{x},t))] + \nabla^2_x: [D(\bm{x},t)p(\bm{x},t)],
\end{equation}
where $D(\bm{x},t)=\frac{1}{2}\sigma^2(\bm{x},t)$ is the diffusion matrix and $\nabla^2_x:$ denotes the tensor dot product. The Master equation~\cite{van1992stochastic} similarly describes probability evolution for discrete state spaces. This creates a duality between the microscopic stochastic process and its macroscopic deterministic evolution.

\subsection{Problem Definition}
Our goal is to reconstruct the continuous evolution trajectory $(\mu_t)_{t\in [0,T]}$ in the space of measures $\mathcal{P}(\mathbb{R}^d)$ from a set of empirical distributions $\{\hat{\mu}_{t_i}\}^N_{i=1}$ observed at $N$ discrete times. 
Each snapshot $\hat{\mu}_{t_i}$ consists of $n$ samples from the system's ensemble. We assume an absolutely continuous path, where each measure $\mu_t$ has a probability density function $p(\bm{x},t)$ satisfying $d\mu_t(\bm{x})=p(\bm{x},t)dx$.
This formulation applies generally to continuous PDF evolutions, including those governed by the Fokker-Planck equation or approximated from Master equations.

We frame this as a dynamical optimal transport problem based on the principle of least action~\cite{koshizuka2023neural, kapusniak2024metric}. 
The optimal path between two consecutive snapshots, $\nu_0$ and $\nu_1$, is found by solving the Benamou-Brenier formulation~\cite{benamou2000computational}:
\begin{equation}~\label{equ:DOT}
\begin{aligned}
    W(\nu_0, \nu_1) = \inf_{p, f} \int_0^1 \int_{\mathbb{R}^d} \frac{1}{2} \|f(\bm{x},t)\|_2^2 p(\bm{x},t) dxdt \\
    \text{s.t.} \quad  \partial_t p + \nabla \cdot (p f) = 0, \quad p|_{t=0} = \nu_0, \quad p|_{t=1} = \nu_1 .
\end{aligned}
\end{equation}
This problem seeks an optimal path, defined by density $p(\bm{x},t)$ and velocity $f(\bm{x},t)$, that minimizes kinetic energy.
\section{Weight Path of Probability Evolution}

In this section, we provide a theoretical justification for modeling stochastic dynamics through the evolution of weights.

\subsection{Parameterizing Probability Distribution}

We parameterize the high-dimensional probability density $p(\bm{x}, t)$ with an autoregressive neural network~\cite{wu2019solving}. 
This factorizes the joint density into a product of conditional probabilities via the chain rule:
\begin{equation}~\label{equ:van}
    p(\bm{x}, t) = p_{\bm{\theta}_t}(\bm{x}) = \prod_{i=1}^d p_{\bm{\theta}_t}(x_i | x_1,...,x_{i-1})
\end{equation}
Here, a complex, high-dimensional distribution is represented by a finite-dimensional weight vector $\bm{\theta}_t$. 
A key advantage of this autoregressive formulation is scalability with dimension $d$, which mitigates the curse of dimensionality. 
The universal approximation property~\cite{hornik1989multilayer} ensures that for any time $t$, a weight set $\bm{\theta}_t$ exists that can accurately represent the true density $p(\bm{x}, t)$.

We formalize this relationship as a map $G$, determined by the network architecture, from the weight space $\Theta$ to the space of probability densities:
\begin{equation}
    G: \bm{\theta}_t \in \Theta \mapsto p_{\bm{\theta}_t}(x) \in \mathcal{P}(\mathbb{R}^d)
\end{equation}
Thus, the deterministic evolution of the probability distribution becomes a trajectory in the network's finite-dimensional weight space. This maps the density path $(p_t)_{t\in [0,T]}$ to a weight path $(\bm{\theta}_t)_{t\in [0,T]}$.

\subsection{Bridging the Probability Path and Evolving Weight}

The dynamic optimal transport problem of equation~(\ref{equ:DOT}) can be approximately solved by finding an optimal trajectory in the network's weight space. 
The core idea is to find a weight-space path that approximates the true optimal path's energy by minimizing a corresponding energy functional. 
We first define this energy functional~\cite{zhang25learning} as:
\begin{equation}~\label{equ:RUOT}
    \mathcal{E}(p, v):=\int_{0}^{1}\int_{\mathbb{R}^d}[\frac{1}{2}||v(x,t)^{2}||+\frac{\sigma^{4}}{8}||\nabla \log p||^{2}]pdxdt,
\end{equation}
where $p$ denotes $p(x,t)$. 
Here, $v(x,t)$ is the velocity field, and the density $p(x,t)$ evolves according to the continuity equation $\partial_t p + \nabla \cdot (p v) = 0$.
The evolution of the parameterizing weights $\bm{\theta}_t$ is then governed by an Ordinary Differential Equation (ODE):
\begin{equation}
    \frac{d \bm{\theta}}{dt} = g(\bm{\theta}, t).
\end{equation}
The weight-space velocity field $g$ defines a corresponding state-space velocity $v_g(x,t)$ that satisfies the continuity equation $\partial_t p_{\bm{\theta}} + \nabla \cdot (p_{\bm{\theta}} v_g) = 0$.
\begin{theorem}~\label{theo:main}
    Let $\mu_0$ and $\mu_1$ be the initial and end probability measures on $\mathbb{R}^d$, and let $(p^*, v^*)$ be the optimal path under the Eq.~\ref{equ:RUOT}.
    Assuming the following conditions are satisfied:
    \begin{enumerate}
        \item[(C1)] There exists a reference parameter $\bm{\theta}_0$ such that $p_{g_0} \approx p^*$ and $v_{g_0} \approx v^*$;
        \item[(C2)] The trained minimizer $\bm{\theta}^*$ satisfies $p_{\bm{\theta}^*}(1) \approx \mu_1$;
        \item[(C3)] $\bm{\theta}^* = \arg\min_{\bm{\theta}} (\lambda\mathcal{E}(p_{\bm{\theta}}, v_g) + \mathcal{L}(g)$) with $\lambda > 0$, where
        \begin{equation}
           \mathcal{E}(p_{\bm{\theta}}, v_g):=\int_{0}^{1}\int_{\mathbb{R}^d}[\frac{1}{2}||v_g||^2+\frac{\sigma^{4}}{8}||\nabla \log p_{\bm{\theta}}||^{2}]p_{\bm{\theta}} dxdt
        \end{equation}
        and $\mathcal{L}(g)=-\mathbb{E}_{x\sim \mu_0}[\log p_{\bm{\theta}_{\mu_0}}(x)]-\mathbb{E}_{x\sim \mu_1}[\log p_{\bm{\theta}_{\mu_1}}(x)]$.
    \end{enumerate}
    Then, we have $|\mathcal{E}(p_{\bm{\theta}^*}, v_{g^*}) - \mathcal{E}(p^*, v^*)| \le \delta$.
\end{theorem}
\begin{proof}
See Appendix for proof.
\end{proof}
Here, $\delta > 0$ is an arbitrarily small error term. A detailed proof is in Appendix~\ref{sec_supp:proof}. Theorem~\ref{theo:main} bridges weight evolution and dynamic optimal transport, showing that the probability path can be projected onto a path in network's weight space.

\subsection{Learning Stochastic Dynamics via Evolving Weight}


We formalize learning stochastic dynamics as an optimization problem in the weight space. 
Given $N$ empirical distributions $\{\hat{\mu}_{t_i}\}^N_{i=1}$, an autoregressive network $G$ (Eq.~\ref{equ:van}) models the state distribution $p(\bm{x}, t)$ by mapping a weight vector $\bm{\theta}_t$ to a density $p_{\bm{\theta}_t}(\bm{x})$. 
Learning the dynamics involves finding an optimal velocity field $g(\bm{\theta}, t)$ whose resulting weight trajectory $\{\bm{\theta}_t\}_{t \in [0,T]}$ minimizes the objective function:
\begin{equation}\label{equ:obj_func}
    \frac{1}{N}\sum^N_{i=1}(-\mathbb{E}_{\bm{x} \sim \hat{\mu}_{t_i}}\log p_{\bm{\theta}_{t_i}}(\bm{x})) + \lambda \mathcal{E}(p_{\bm{\theta}}, v_g).
\end{equation}
The objective's first term ensures data fit, while the second constrains the path energy.
\section{WeightFlow}

In this section, we introduce WeightFlow, a novel framework for modeling stochastic dynamics through the evolution of neural network weights, as illustrated in Figure~\ref{fig:framework}.

\begin{figure*}[!ht]
    \centering
    \includegraphics[width=0.85\linewidth]{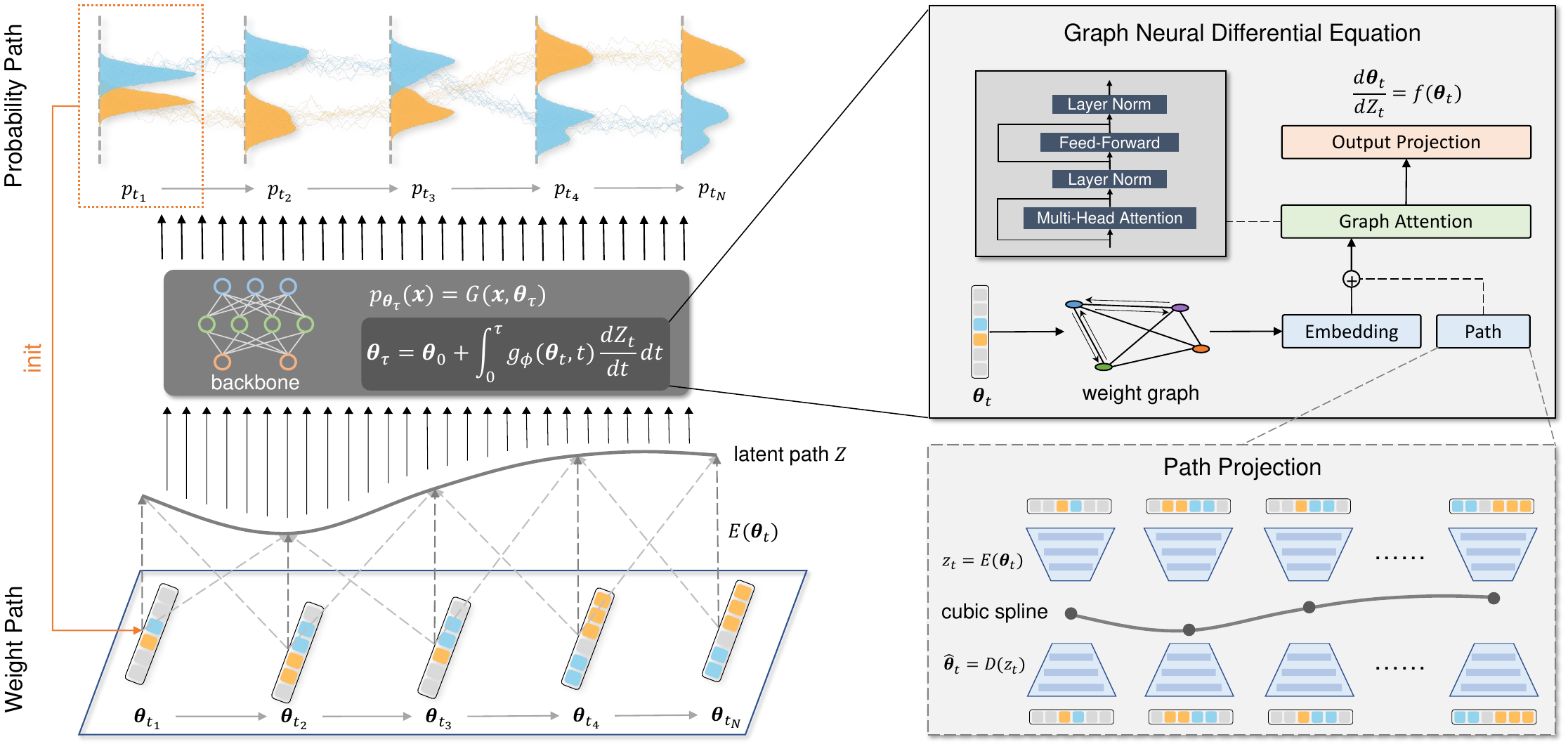}
    \caption{Framework of WeightFlow.}
    \label{fig:framework}
\end{figure*}

\subsection{Backbone and Weight Graph}

WeightFlow uses an autoregressive model, termed the \textbf{backbone}, to parameterize the state distribution $p(\bm{x}, t)$. 
The backbone implements the factorization in Eq.~\ref{equ:van} and can be any sequential architecture like an RNN or Transformer. 
For an RNN-based backbone, it conditions the $i$-th dimension's probability on a hidden state $h_{i-1}$: $p(x_i,t) = p(x_i|h_{i-1}, \bm{\theta}_t)$.
Normalization of $p(\bm{x}, t)$ is ensured by applying a softmax output to each conditional density $p(x_i,t)$.
We train a separate backbone for each time snapshot $t_i$ by minimizing the data's Negative Log-Likelihood (NLL). 
This yields a set of weights $\{\bm{\theta}_{t_i}\}_{i=1}^N$ that serve as anchor points for learning the continuous dynamics.
Details of the Transformer-based backbone are provided in Appendix~\ref{sec_supp:transformer}.

The weights of the backbone are inherently structured by the network's architecture. 
To represent this structure, we organize the weights of each snapshot $\bm{\theta}_{t_i}$ into a \textbf{weight graph} based on the network's forward data flow. 
Specifically, each output neuron of the backbone's linear layers becomes a node in the graph. 
Each node's feature vector is formed by concatenating its incoming connection weights and its bias term. 
Thus, a linear layer with weights $w \in \mathbb{R}^{D_{out} \times D_{in}}$ and bias $b \in \mathbb{R}^{D_{out} \times 1}$ becomes $D_{out}$ nodes in the weight graph, each with a $(D_{in}+1)$-dimensional feature. 
Weight evolution is thereby formalized as the evolution of node features on this graph.

\subsection{Graph Neural Differential Equation}

To learn the optimal weight path from Eq.~\ref{equ:obj_func}, we parameterize the velocity field $g$ as a neural differential equation. 
Specifically, WeightFlow models the evolution of the backbone weights $\bm{\theta}_t$ as a graph neural ODE~\cite{chen2018neural}:
\begin{equation}\label{equ:ode}
    \bm{\theta}_\tau = \bm{\theta}_0 + \int^\tau_0 g_{\phi}(\bm{\theta}_t, t) dt,
\end{equation}
where $g_{\phi}=\frac{d{\bm{\theta_t}}}{dt}$ is a hypernetwork modeling the dynamics of the backbone weights $\bm{\theta}_t$. 
To handle the heterogeneous dimensions of the weight graph's node features, we first use a layer-wise linear map to project them into a common dimension. 
We then apply a multi-head attention mechanism to the fully connected weight graph to model all inter-node relationships. 
The resulting representation is projected back to the original dimensions via another layer-wise map to predict the weight derivative.

To capture the global path's manifold trend, we enhance the ODE into a Controlled Differential Equation (CDE)~\cite{kidger2020neural} via a Riemann–Stieltjes integral:
\begin{equation}
\bm{\theta}_\tau = \bm{\theta}_0 + \int^\tau_0 g_{\phi}(\bm{\theta}_t, t) \frac{dZ_t}{dt} dt,
\end{equation}
where integration follows a controlling variable path $Z_t$. 
A self-supervised autoencoder (encoder $E$, decoder $D$) projects the anchor weights $\{\bm{\theta}_{t_i}\}_{i=1}^N$ into latent vectors $\{z_{t_i}\}_{i=1}^N$. 
We then use cubic spline interpolation on these latent vectors to approximate the path's rate of change, $\frac{dZ_t}{dt}$. 
This injects the global evolutionary trend into the integration and explicitly scales the conditional vector field $g_\phi=\frac{d{\bm{\theta_t}}}{dZ_t}$.

\subsection{Training Strategy}

WeightFlow's training has two stages: anchor pre-training and dynamics model training.
First, in the warm-up stage, we pre-train the backbone independently at each observation time $t_i$ to get the anchor weights $\{\bm{\theta}_{t_i}\}_{i=1}^N$. 
To ensure a smooth transition, we use a sequential-aligning strategy, initializing the training for time $t_i$ with the weights from $t_{i-1}$. 
These anchor weights are projected into a latent space and interpolated with cubic splines to form the control path $Z_t$.
The main stage optimizes the hypernetwork $g_\phi$ that governs weight dynamics. 
In the objective from Eq.~\ref{equ:obj_func}, the NLL loss is replaced with a weight reconstruction loss. 
This reconstruction loss is a mean squared error between the trajectory's weights and the pre-trained anchor weights $\bm{\hat{\theta}}_{t_i}$. 
Path energy is computed via a ODE solver.
We detail the training procedure in Algorithm~\ref{alg:training} of the Appendix.

\subsection{Time and Computational Complexity}

Let $L$ be the number of candidate states or mixture density parameters per dimension in a $d$-dimensional system. WeightFlow's complexity depends on its autoregressive backbone and the hypernetwork $g_{\phi}$.
\begin{itemize}
    \item \textbf{Backbone}: The backbone's size is independent of dimension $d$, with $O(L)$ space complexity. Its inference time complexity is $O(d)$ due to the autoregressive calculation.
    \item \textbf{Hypernetwork}: The hypernetwork's inference bottleneck is its multi-head attention mechanism. Its time complexity is $O(N_{nodes}^2)$, where the number of nodes, $N_{nodes}$, is determined by the backbone architecture (e.g., proportional to $L$) and is also independent of dimension $d$.
\end{itemize}
Thus, WeightFlow effectively avoids the curse of dimensionality, as its main computational costs do not scale with the state-space dimension.
\section{Numerical Results}
In this section, we empirically evaluate WeightFlow on a diverse set of simulated and real-world stochastic dynamics.

\begin{table*}[!ht]
    \centering
    \setlength{\tabcolsep}{1pt}

    \begin{tabular}{l *{10}{c}}
        \toprule
        & \multicolumn{2}{c}{Epidemic} & \multicolumn{2}{c}{Toggle Switch} & \multicolumn{2}{c}{Signalling Cascade1} & \multicolumn{2}{c}{Signalling Cascade2} & \multicolumn{2}{c}{Ecological Evolution} \\

        
        \cmidrule(lr){2-3} \cmidrule(lr){4-5} \cmidrule(lr){6-7} \cmidrule(lr){8-9} \cmidrule(lr){10-11}

        $\times 10^{-1}$ & $\mathcal{W} \downarrow $ & $JSD \downarrow$ & $\mathcal{W} \downarrow$ & $JSD \downarrow$ & $\mathcal{W} \downarrow$ & $JSD \downarrow$ & $\mathcal{W} \downarrow$ & $JSD \downarrow$ & $\mathcal{W} \downarrow$ & $JSD \downarrow$ \\
        \midrule
        
        Latent SDE    & $3.14_{\pm0.25}$ & $4.22_{\pm0.26}$ & $2.34_{\pm0.15}$ & $1.27_{\pm0.12}$ & $3.04_{\pm0.17}$ & $0.85_{\pm0.14}$ & $3.59_{\pm0.13}$ & $1.02_{\pm0.06}$ & $8.04_{\pm0.33}$ & $3.52_{\pm0.23}$ \\
        Neural MJP    & $1.88_{\pm0.14}$ & $1.61_{\pm0.14}$ & $2.13_{\pm0.26}$ & $0.94_{\pm0.14}$ & $1.69_{\pm0.15}$ & $0.30_{\pm0.04}$ & $1.68_{\pm0.11}$ & $0.36_{\pm0.01}$ & $1.68_{\pm0.18}$ & $0.51_{\pm0.03}$ \\
        T-IB          & $2.62_{\pm0.17}$ & $3.52_{\pm0.29}$ & $1.59_{\pm0.20}$ & $0.88_{\pm0.11}$ & $1.66_{\pm0.16}$ & $0.32_{\pm0.04}$ & $2.16_{\pm0.17}$ & $0.40_{\pm0.03}$ & $2.17_{\pm0.24}$ & $0.56_{\pm0.06}$ \\
        NLSB          & $3.27_{\pm0.28}$ & $1.65_{\pm0.14}$ & $2.97_{\pm0.30}$ & $1.32_{\pm0.20}$ & $1.50_{\pm0.10}$ & $0.39_{\pm0.05}$ & $1.83_{\pm0.15}$ & $0.48_{\pm0.05}$ & $3.09_{\pm0.26}$ & $2.80_{\pm0.32}$ \\
        DeepRUOT      & $1.78_{\pm0.13}$ & $1.08_{\pm0.09}$ & $1.37_{\pm0.17}$ & $0.77_{\pm0.05}$ & $\underline{0.52_{\pm0.02}}$ & $0.07_{\pm0.00}$ & $\underline{0.51_{\pm0.01}}$ & $0.08_{\pm0.00}$ & $3.27_{\pm0.31}$ & $2.47_{\pm0.36}$ \\
        
        \midrule
        WeightFlow$_{\text{O}}$        & $\underline{1.14_{\pm0.15}}$ & $\underline{0.36_{\pm0.02}}$ & $\underline{0.90_{\pm0.08}}$ & $\underline{0.35_{\pm0.02}}$ & $0.59_{\pm0.06}$ & $\underline{0.05_{\pm0.00}}$ & $0.64_{\pm0.08}$ & $\underline{0.07_{\pm0.01}}$ & $\mathbf{0.50_{\pm0.08}}$ & $\underline{0.13_{\pm0.01}}$ \\
        WeightFlow$_{\text{C}}$ & $\mathbf{1.10_{\pm0.14}}$ & $\mathbf{0.34_{\pm0.01}}$  & $\mathbf{0.82_{\pm0.07}}$ & $\mathbf{0.33_{\pm0.02}}$ & $\mathbf{0.48_{\pm0.03}}$ & $\mathbf{0.04_{\pm0.00}}$ & $\mathbf{0.49_{\pm0.07}}$ & $\mathbf{0.06_{\pm0.01}}$ & $\underline{0.51_{\pm0.07}}$ & $\mathbf{0.12_{\pm0.02}}$ \\
        Promotion  & 38.20\% & 68.52\% & 40.15\% & 57.14\% & 7.69\% & 42.86\% & 3.92\% & 25.00\% & 70.24\% & 76.47\%  \\
        
        \bottomrule
    \end{tabular}
    \caption{Statistic results on various stochastic dynamical systems over 10 runs. The best/second-best are bold/underlined.}
    \label{tab:main_results}
\end{table*}

\begin{figure*}[!ht]
    \centering
    \includegraphics[width=0.9\linewidth]{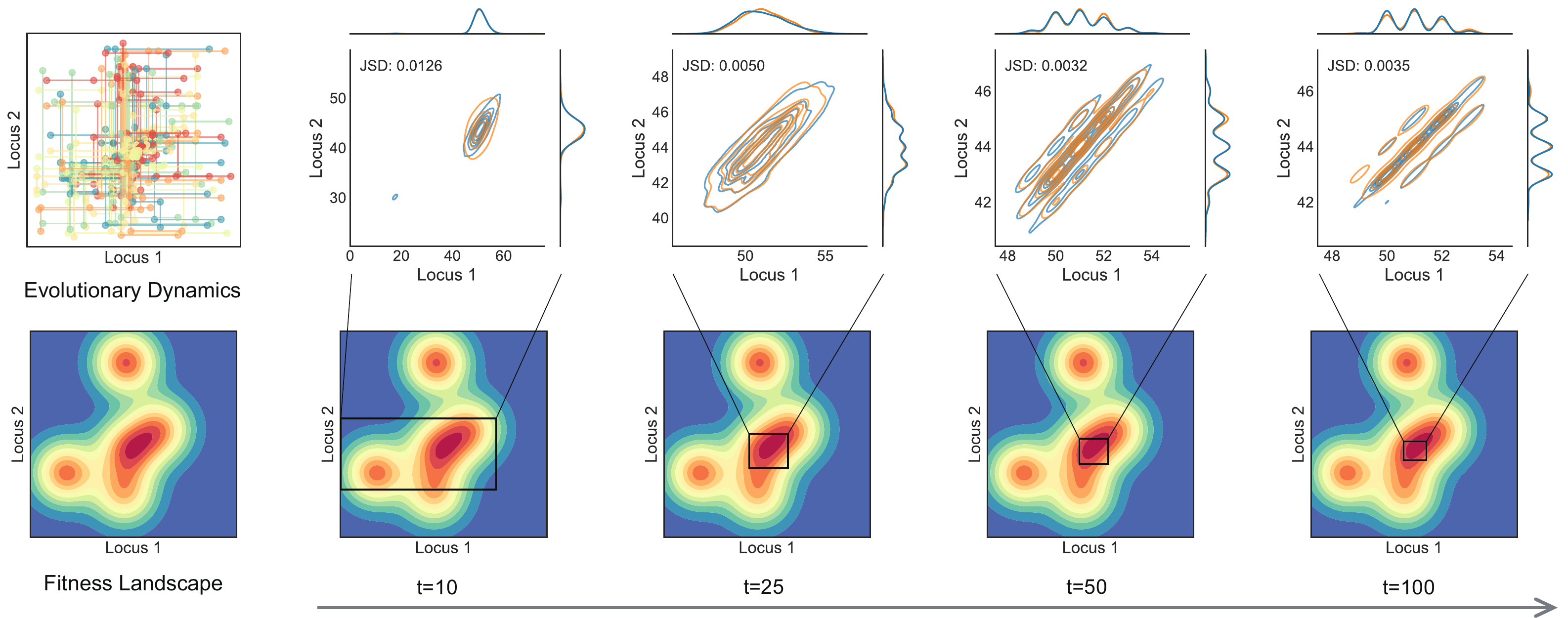}
    \caption{Joint and marginal distributions predicted by WeightFlow over time on the Ecological Evolution system. The central panels show the joint distribution contours of two loci overlaid on the fitness landscape (brighter yellow indicates higher fitness), with marginal densities on the outer axes.}
    \label{fig:main_sswm}
\end{figure*}

\subsection{Experimental Settings}

\subsubsection{Dynamical Systems}

We evaluate WeightFlow's capacity to model interdisciplinary stochastic dynamics in both discrete and continuous state spaces across a diverse set of simulated and real-world systems.
The simulated systems include four biochemical reaction networks~\cite{tang2023neural}: Epidemic, Toggle Switch, and Signalling Cascades 1 and 2, as well as one ecological adaptive evolution process~\cite{li2025predicting}. 
The real-world systems consist of two single-cell differentiation datasets: the human embryoid body~\cite{tong2020trajectorynet} and the \textit{in vitro} pancreatic $\beta$-cell~\cite{veres2019charting}. 
Detailed system dynamics, data collection for each dataset are provided in the Appendix.

\subsubsection{Baselines}

We benchmark WeightFlow against several state-of-the-art baselines that model stochastic dynamics from diverse perspectives. 
These include methods based on the direct modeling of stochastic differential equations, LatentSDE~\cite{li2020scalable}; Markov jump processes, NeuralMJP~\cite{seifner2023neural}; transfer operator, T-IB~\cite{federici2024latent}; Schrödinger bridges, NLSB~\cite{koshizuka2023neural}; and dynamic optimal transport, DeepRUOT~\cite{zhang25learning}. 
The configurations are detailed in the Appendix~\ref{sec_supp:baselines}.

\subsubsection{Metrics}
We measure the difference between probability densities using several standard metrics from related work. These include the Wasserstein ($W$) distance~\cite{zhang25learning}, Maximum Mean Discrepancy ($MMD$)~\cite{kapusniak2024metric}, and the Jensen-Shannon Divergence ($JSD$)~\cite{li2025predicting}. 
The detailed computational formulas for these metrics are provided in the Appendix~\ref{sec_supp:metrics}.

\subsection{Main Results}

In all experiments, unless otherwise specified, WeightFlow uses a RNN with Gated Recurrent Units (GRU) as its default backbone. 
The subscripts $_O$ and $_C$ denote implementations based on ODEs and CDEs, respectively. 
When not explicitly subscripted, WeightFlow defaults to the CDE-based implementation. 
The specific configuration of WeightFlow for each system is provided in the Appendix~\ref{sec_supp:weightflow_conf}.

\begin{figure*}[!ht]
    \centering
    \includegraphics[width=0.95\linewidth]{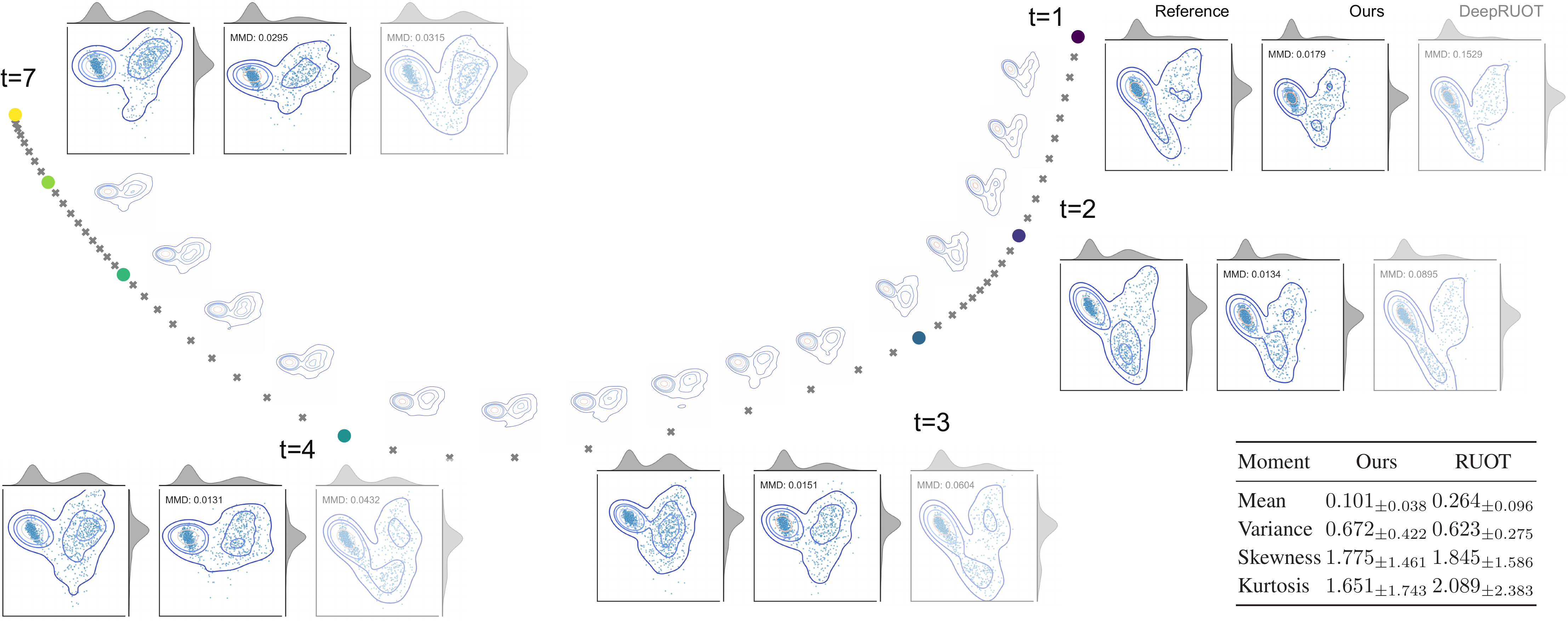}
    \caption{Weight prediction for $\beta$-cell differentiation. The trajectory shows the continuous evolution of weights and corresponding ensemble distributions (PCA). Highlights compare the results at observed snapshots: reference (left), WeightFlow's prediction (center), and DeepRUOT's prediction (right). The table reports the average relative error of the first four statistical moments.}
    \label{fig:veres}
\end{figure*}

\subsubsection{Simulated Datasets}

For five discrete systems, we simulate 1,000 evolution trajectories of 100 steps each. 
Training data is created by down-sampling each trajectory to 10 time points, while evaluation is performed on the full 100 time points.

Across all systems, WeightFlow outperforms baselines, improving the Wasserstein distance and Jensen-Shannon divergence by 32.04\% and 53.99\% on average, respectively (Table~\ref{tab:main_results}). 
These results demonstrate WeightFlow's ability to capture probabilistic dynamics and validate our weight-space modeling approach. 
Compared to DeepRUOT, WeightFlow provides more stable long-term predictions by avoiding late-stage error accumulation (Appendix Figure~\ref{fig_supp:main_w1_curves}).

A case study on an ecological evolution system further demonstrates WeightFlow's performance in Figure~\ref{fig:main_sswm}. 
In this system, a 2D genetic phenotype evolves towards a global peak on a fitness landscape. 
WeightFlow accurately predicts the distribution throughout the evolution, capturing both macroscopic landscape shifts and fine-grained local dynamics.

\begin{table}[t]
    \centering
    \setlength{\tabcolsep}{1pt}
    \resizebox{\linewidth}{!}{
    \begin{tabular}{lcccc}
        \toprule
         & \multicolumn{2}{c}{$\beta$-cell} & \multicolumn{2}{c}{Embryoid} \\
         \cmidrule(lr){2-3} \cmidrule(lr){4-5}
         
         & $\mathcal{W} \downarrow$ & $MMD \downarrow$ & $\mathcal{W} \downarrow$ & $MMD \downarrow$ \\
         \midrule
         NLSB    & $11.18_{\pm0.22}$ & $0.07_{\pm0.01}$ & $14.39_{\pm0.40}$ & $0.10_{\pm0.03}$ \\
         RUOT    & $10.99_{\pm0.20}$ & $0.06_{\pm0.01}$ & $14.71_{\pm0.49}$ & $0.15_{\pm0.03}$ \\
         WeightFlow$_{\text{O}}$ & $\underline{9.86_{\pm0.25}}$ & $\underline{0.02_{\pm0.01}}$ & $\mathbf{13.72_{\pm0.39}}$ & $\mathbf{0.02_{\pm0.01}}$\\
         WeightFlow$_{\text{C}}$ & $\mathbf{9.73_{\pm0.27}}$ & $\mathbf{0.02_{\pm0.01}}$ & $\underline{14.18_{\pm0.43}}$ & $\underline{0.03_{\pm0.01}}$\\
         Promotion     & 11.44\% & 71.88\% & 4.64\% & 81.37\% \\
         \bottomrule
    \end{tabular}
    }
    \caption{Statistical results on real-world cell datasets.}
    \label{tab:scRNA}
\end{table}

\begin{figure}[ht]
    \centering
    \includegraphics[width=0.98\linewidth]{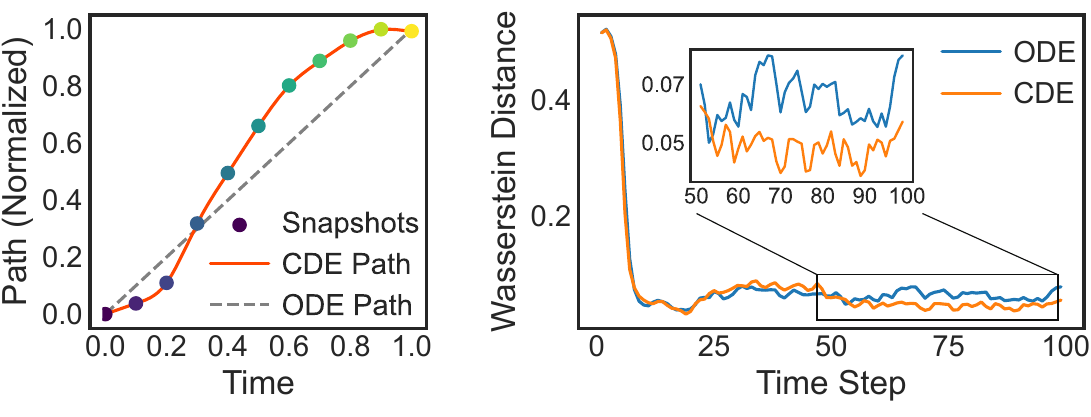}
    \caption{Comparing the path (left) and average error (right) for ODE and CDE models on the Toggle Switch system.}
    \label{fig:path}
\end{figure}

\subsubsection{Real-world Datasets}

We evaluate WeightFlow on two continuous-space, single-cell differentiation datasets. 
The datasets track gene expression via scRNA-seq during \textit{in vitro} pancreatic $\beta$-cell and human embryoid body differentiation. 
Following related work~\cite{koshizuka2023neural,zhang25learning}, we project the high-dimensional gene data onto 30 and 100 principal components. 
We then use a Mixture Density Network (MDN) as WeightFlow's backbone to model the resulting continuous distributions.
The detailed configuration is provided in the Appendix~\ref{sec_supp:van_conf}.

WeightFlow significantly outperforms baselines on these systems, demonstrating its effectiveness in continuous state spaces (Table~\ref{tab:scRNA}). 
A case study of the pancreatic $\beta$-cell visualizes the predicted paths in Figure~\ref{fig:veres}. 
We compare first four statistical moments of the predicted distributions against the ground truth. 
While both WeightFlow and baseline capture the global shape with comparable errors in mean and variance, WeightFlow is significantly more accurate for higher-order moments like skewness and kurtosis. 
This shows WeightFlow accurately reproduces fine-grained distribution structures like asymmetry and extreme values. 
Similar results are observed for the embryoid body dataset in the Appendix Figure~\ref{fig_supp:embryoid}.

\begin{figure*}[t]
    \centering
    \begin{subfigure}[t]{0.32\textwidth}
        \includegraphics[width=\textwidth]{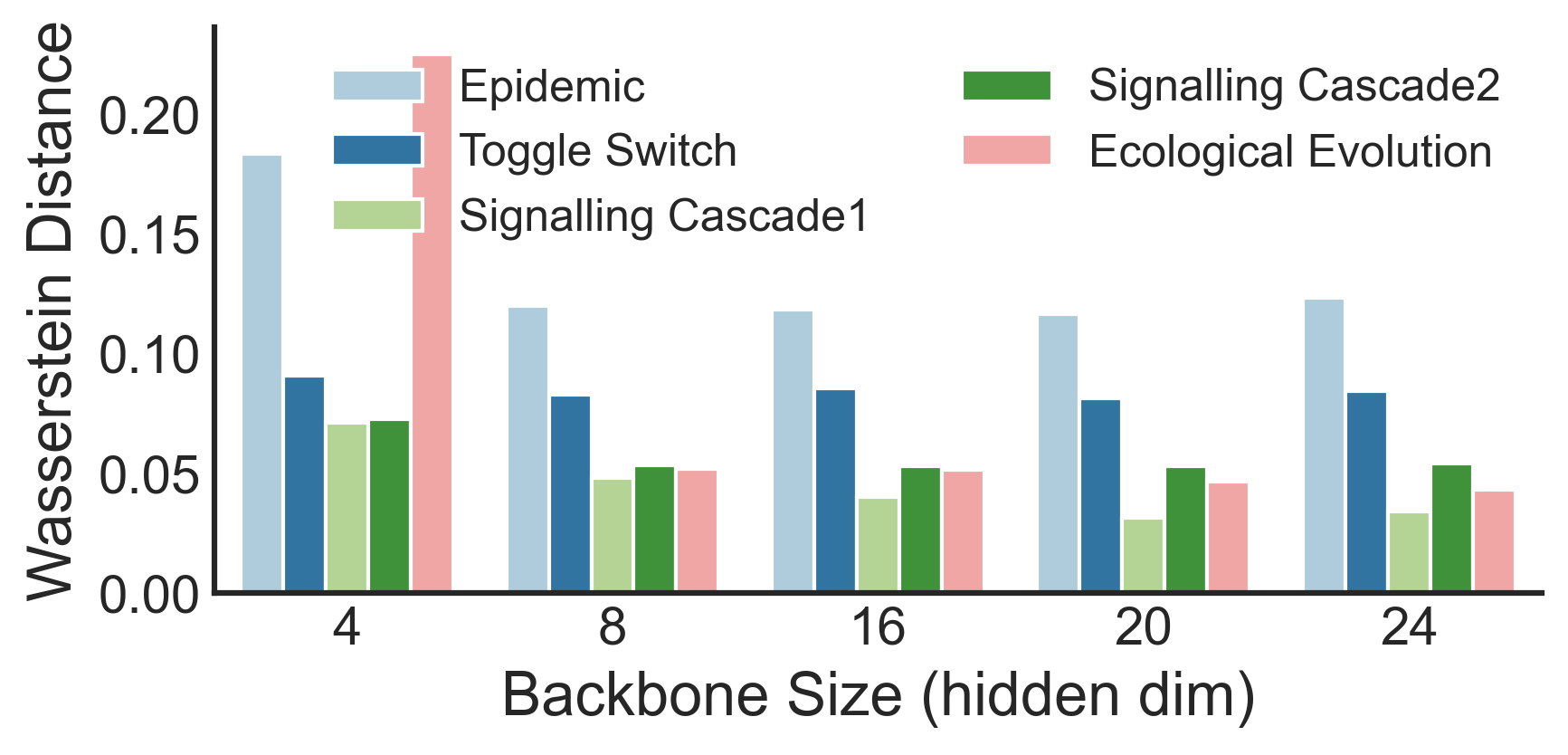}
        \caption{Backbone Size}
    \end{subfigure}
    \hfill
    \begin{subfigure}[t]{0.32\textwidth}
        \includegraphics[width=\textwidth]{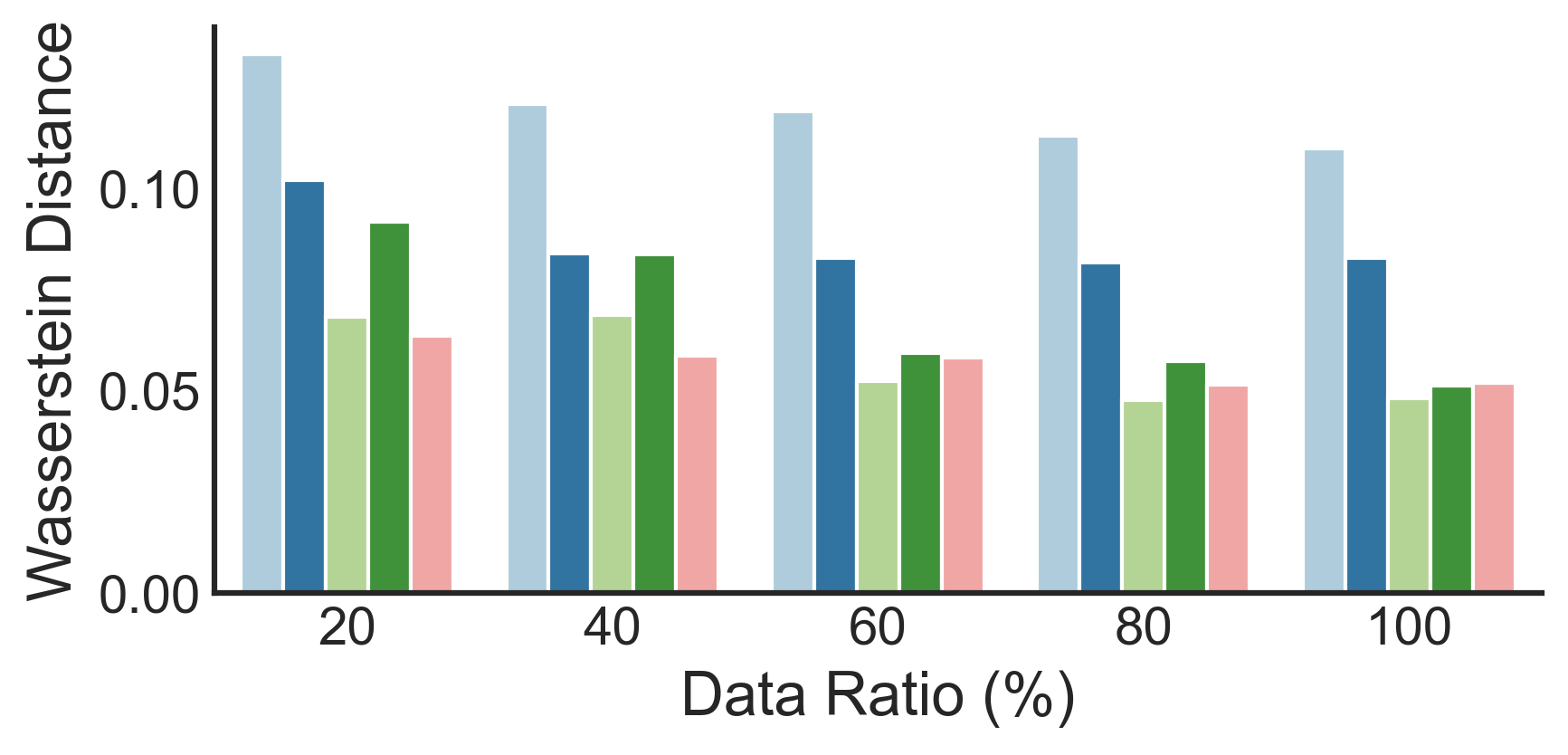}
        \caption{Data Ratio}
    \end{subfigure}
    \hfill
    \begin{subfigure}[t]{0.32\textwidth}
        \includegraphics[width=\textwidth]{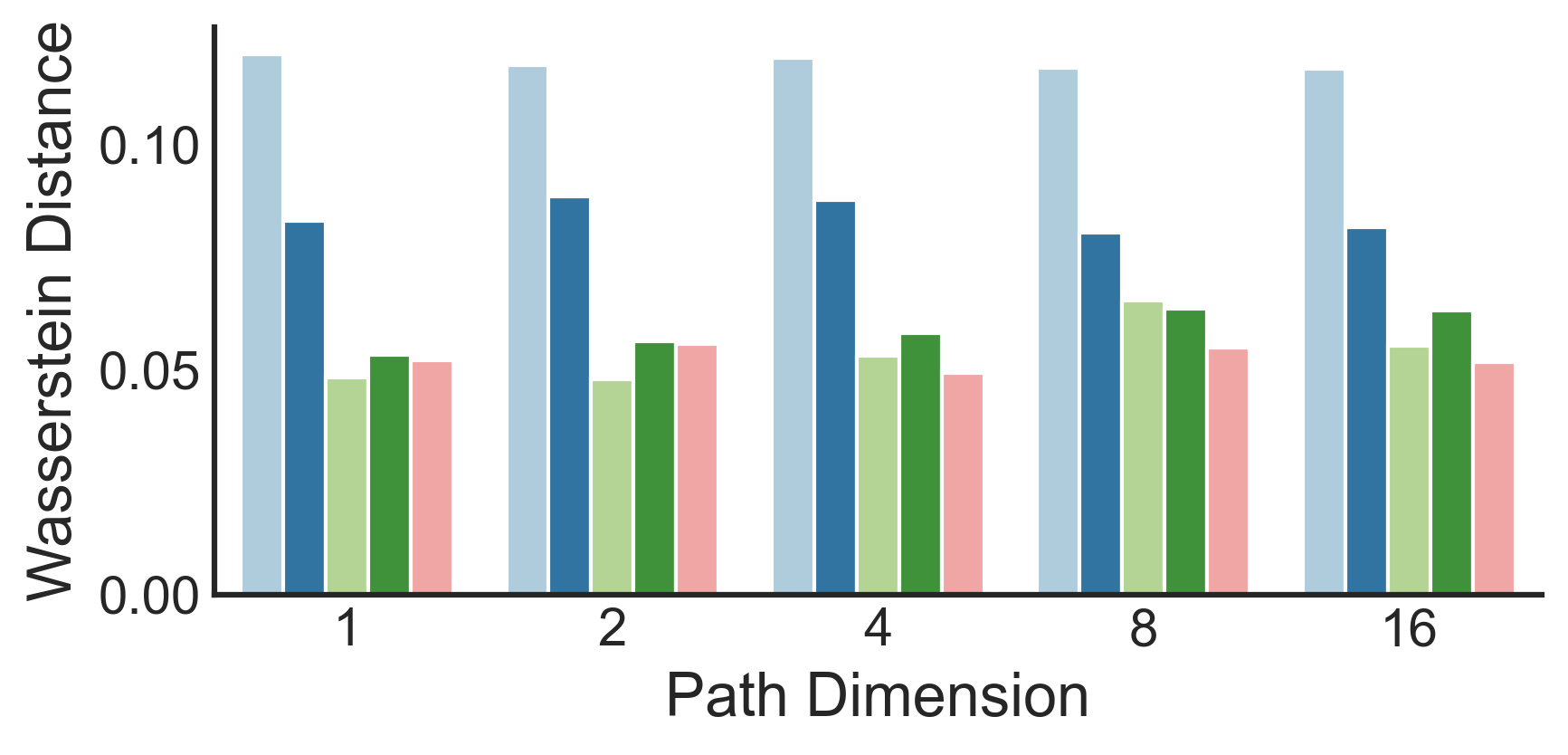}
        \caption{Path Dimension}
    \end{subfigure}

    \caption{Average Wasserstein distance for WeightFlow as a function of (a) backbone size, (b) data ratio, and (c) path dimension across five dynamical systems.}
    \label{fig:robustness}
\end{figure*}

\begin{figure*}[t]
    \centering

    \begin{subfigure}[b]{0.315\linewidth}
        \centering
        \includegraphics[width=\linewidth]{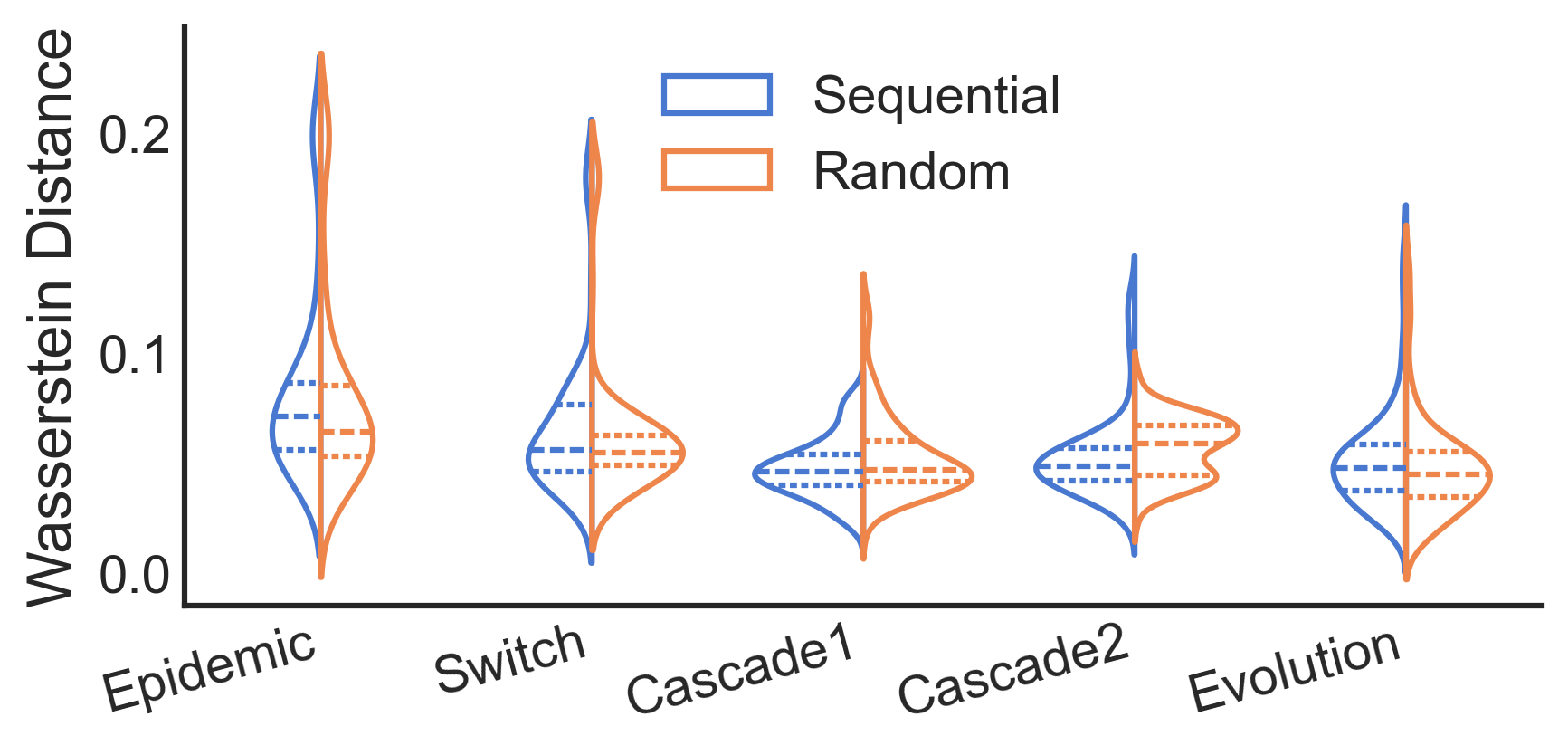}
        \caption{Autoregressive Order}
    \end{subfigure}
    \hfill
    \begin{subfigure}[b]{0.315\linewidth}
        \centering
        \includegraphics[width=\linewidth]{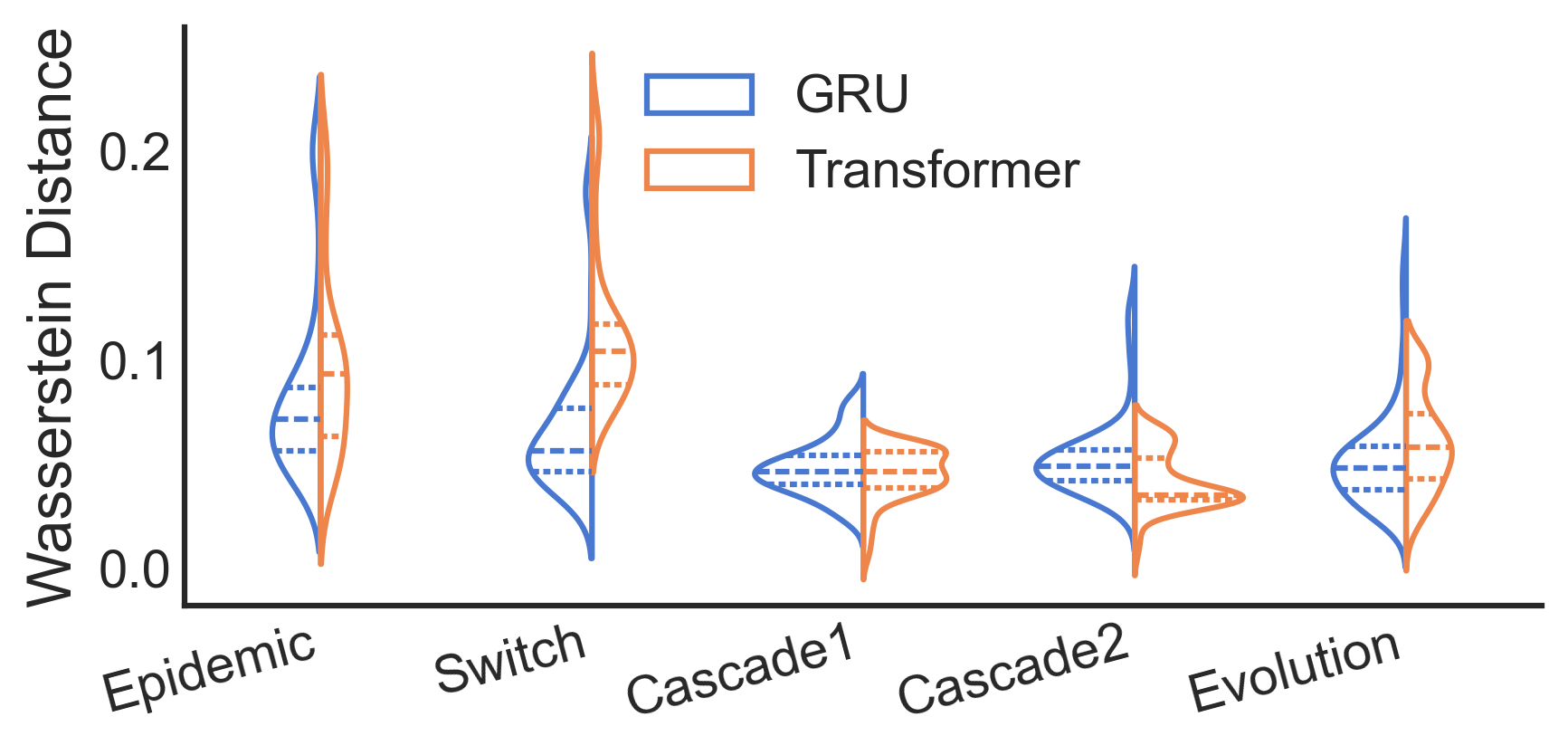}
        \caption{Backbone Architecture}
    \end{subfigure}
    \hfill
    \begin{subfigure}[b]{0.315\linewidth}
        \centering
        \includegraphics[width=\linewidth]{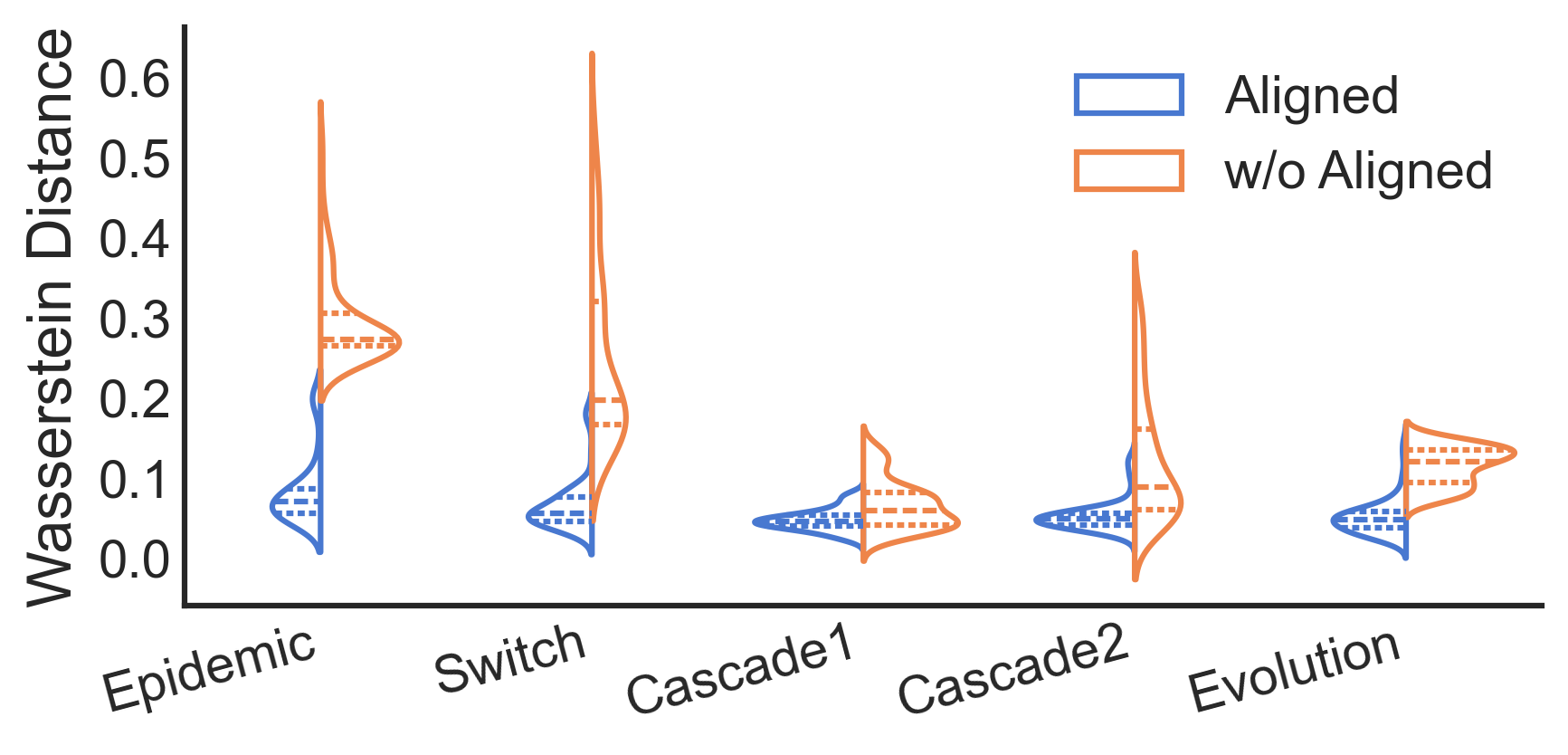}
        \caption{Sequential Aligning}
    \end{subfigure}

    \caption{Wasserstein distance distributions for WeightFlow across five systems under various settings: (a) autoregressive order, (b) backbone architecture, and (c) sequential aligning (with versus without). The outliers are truncated to highlight differences.}
    \label{fig:backbone_order}
\end{figure*}

\begin{figure}[t]
    \centering
    \includegraphics[width=\linewidth]{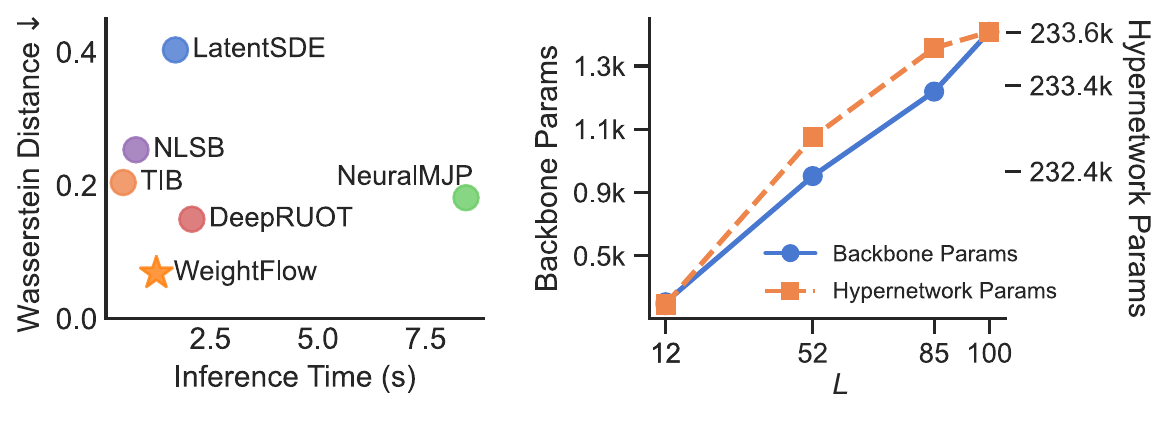}
    \caption{Inference time (left) and model size (right).}
    \label{fig:time_space}
\end{figure}

\subsection{Ablation Study}

\subsubsection{Path Integral}
WeightFlow$_{\text{C}}$, which integrates along the latent path $Z$, generally outperforms WeightFlow$_{\text{O}}$, which integrates directly with respect to time $t$ (Table~\ref{tab:main_results}). 
Taking the Toggle Switch system as example (Figure~\ref{fig:path}), the integration path of WeightFlow$_{\text{C}}$ is able to capture the slow-down behavior during the later convergence phase, which is consistent with the system evolution trend. 
This, in turn, guides the weight evolution to decelerate in later stages, which corresponds to WeightFlow$_{\text{C}}$ having lower long-term errors. 
The results for other systems are in the Appendix Figure~\ref{fig_supp:path_integral}.

\subsubsection{Path Projection}
We replace the autoencoder with different projection algorithms for computing the control variable $Z_t$. 
The results indicate that WeightFlow consistently achieves the best performance when using the autoencoder-based path projection (Appendix Table~\ref{tab_supp:path_method}). 
Visualizations of different paths are provided in the Appendix Figure~\ref{fig_supp:path_projection}.

\subsubsection{Sequential Aligning}
We evaluate the importance of the sequential warm-start strategy by disabling it during the warm-up stage and instead using random initializations at each observation time. 
The resulting performance degradation confirms the importance of the sequential strategy (Figure~\ref{fig:backbone_order}c).

\subsection{Sensitivity Analysis}


\subsubsection{Backbone Size}
WeightFlow uses the backbone to parameterize probability distributions. Therefore, the backbone's size affects the quality of its fit to complex distributions. 
We test the performance of a single-layer RNN with different latent dimensions (Figure~\ref{fig:robustness}a). 
For these systems, a hidden dimension of just 8 is sufficient ($< 1k$ parameters). 
This benefits from WeightFlow's decoupled design, where the backbone only models static distributions, leaving the temporal evolution to the hypernetwork.

\subsubsection{Data Size}
We evaluate WeightFlow's dependency on the amount of observational data by reducing the number of training samples (Figure~\ref{fig:robustness}b). 
The results show that WeightFlow's performance does not collapse even when the data is reduced to only 20\%, but instead shows only a limited decline.

\subsubsection{Path Dimension}
The dimension of the latent path affects WeightFlow's capacity to represent weight dynamics. 
We find that WeightFlow is not sensitive to the path dimension (Figure~\ref{fig:robustness}c). 
A 1-dim path is even slightly superior to that with higher dimensions. 
This could be because the weight evolution dynamics occur on a low-dimensional manifold (Figure~\ref{fig:veres}), for which 1-dim path is sufficient.

\subsubsection{Autoregressive Order}
Mathematically, WeightFlow's autoregressive backbone models the joint probability distribution in the form of Eq.~\ref{equ:van}, which supports any autoregressive order. 
We test the performance with a randomly shuffled autoregressive order (Figure~\ref{fig:backbone_order}a). 
The results indicate that the prediction error of the random order is indeed very close to that of the sequential order.

\subsubsection{Backbone Architecture}
WeightFlow supports any implementation for its backbone. 
Therefore, we compare the performance of a Transformer-based WeightFlow (Figure~\ref{fig:backbone_order}b). 
The results show that the predictive performance of WeightFlow is similar for both backbone architectures.

\subsection{Time and Space Cost}

We report the inference time versus error for various methods, as well as the parameter size of WeightFlow in Figure~\ref{fig:time_space}. 
WeightFlow achieves the Pareto frontier in terms of inference time and error. 
Furthermore, its parameter size scales only linearly with the number of candidate states, $L$.
\section{Conclusions}

We propose WeightFlow, a new framework for modeling stochastic dynamics in the weight space of a neural network. 
By learning the continuous evolution of a weight graph, WeightFlow effectively circumvents the curse of dimensionality in high-dimensional state spaces.
We find WeightFlow exhibits good robustness to data quantity and model size, and that a very low-dimensional latent path (even 1-dim) is sufficient to capture the complex dynamics of weight evolution. 
This suggests the evolution of the ensemble distribution may lie on a concise manifold. 
Furthermore, compared to baselines, WeightFlow more accurately captures the higher-order moments, indicating a stronger ability to characterize fine structures.
This indicates WeightFlow is a promising solution for an effective solution to stochastic dynamics modeling.

\section{Acknowledgments}
This work is supported in part by the National Key Research and Development Program of China under 2024YFC3307603, and the National Natural Science Foundation of China under 92270114 and 62171260. 

\bibliography{z_reference}

\end{document}